\documentclass[11pt]{article}
\usepackage{amsmath,amsthm,amssymb,mathrsfs,graphicx}
\numberwithin{equation}{section}

\newtheorem{thm}{Theorem}              
\newtheorem{lemma}[thm]{Lemma}
                                       
\def\real{{\bf R}}
\def\S{{\bf S}}
\def\k{\breve{\kappa}}
\def\ie{{\it i.e. }}
\def\eg{{\it e.g. }}

\begin{document}
\twocolumn

\title{A geometric invariant for the study of planar curves and its application
to spiral tip meander.} 
\author{Copyright \copyright \; 2016 Scott Hotton \\ 
email: {\tt scotton@sdf.org}}
\date{}
\maketitle

\tableofcontents
\section{Introduction}

\nopagebreak

\section{The geometric invariant $|\k|$}
\label{S:kappabreve}

   Lets begin by reviewing some differential geometry of planar curves.  We 
let $t$ stand for time and we denote the position of the moving point at time 
$t$ by $(x(t), y(t))^T$.  There are two important functions associated to twice 
differentiable planar curves, their speed and local curvature (more commonly 
referred to as just curvature).  We denote the speed by $v(t)$ and the local 
curvature by $\kappa(t)$.  Also we will denote the direction of the velocity by 
$\theta \in \S$.  Roughly speaking the local curvature tells us how fast 
$\theta$ is changing at a point of the curve.  

   If the velocity is defined and never equal to $(0,0)^T$ the curve is said to 
be {\it immersed}.  If $v(t) \equiv 1$ the curve is said to have {\it unit 
speed}.  In theory immersed curves can be reparameterized to have unit speed.  
This is done using the concept of arc length.  A closed form for arc length can 
sometimes be obtained from the integral
\begin{equation*}
s(t) = \int_0^t v(\tau) \; d\tau 
= \int_0^t \sqrt{\dot{x}(\tau)^2+\dot{y}(\tau)^2}\; d\tau
\end{equation*}
although in actual practice its often not feasible to find an anti-derivative 
for $\sqrt{\dot{x}(\tau)^2+\dot{y}(\tau)^2}$ because of the square root.  In 
any case local curvature has been defined as $\kappa(s) = d\theta/ds$.  

   Although its usually not feasible to compute $\kappa(s)$ it is usually not
too difficult to compute
\begin{equation*}
\kappa(s(t)) = \frac{ \dot{x}(t)\ddot{y}(t) - \dot{y}(t)\ddot{x}(t)}
                           {(\dot{x}(t)^2 + \dot{y}(t)^2)^{3/2}} 
\end{equation*}
and this is often sufficient for many purposes.  Its common practice to write 
$\kappa(t)$ for $\kappa(s(t))$ and we will use this convention here.  

   Roughly speaking the total curvature of a curve is how much $\theta$ changes
over the whole curve.  Historically total curvature has been defined as the 
integral of local curvature.  This had the drawback of making it seem that 
local curvature needed to be defined in order for total curvature to be 
defined.  Fox and Milnor realized however that total curvature is a meaningful 
concept for all geometric curves \cite{fox,sullivan}.  Moreover wherever local 
curvature is a meaningful concept it can be defined in terms of total curvature 
so total curvature is the more fundamental concept.  Total curvature can still 
be computed from the integral of local curvature when local curvature is 
defined but this is no longer regarded as a definition of total curvature.  
This is now known as the Fox-Milnor theorem and it is how we will compute total
curvature here.

   Note that the Fox-Milnor theorem computes total curvature as the integral of
$\kappa(s)$ with respect to arc length not as the integral of $\kappa(t)$ with 
respect to time.  By the change of variables theorem
\begin{equation*}
\int_{s_1}^{s_2} \kappa(\sigma) \; d\sigma
= \int_{t_1}^{t_2} \kappa(\tau) v(\tau) \; d\tau
\neq \int_{t_1}^{t_2} \kappa(\tau) \; d\tau
\end{equation*}

   The integral for total curvature could be reformulated using 
$\dot{\theta}(t) = \kappa(t) v(t)$.  However if we apply the fundamental 
theorem of calculus naively we might write
\begin{equation*}
\int_{t_1}^{t_2} \dot{\theta}(\tau) \; d\tau = \theta(t_2) - \theta(t_1)
\end{equation*}
but this is only the difference between the starting direction and final 
direction of the velocity and it overlooks the possibility that the velocity 
may have undergone several complete turns during the time interval
$[t_1,t_2]$.  To take this possibility into consideration we define the 
function.
\begin{equation*}
\varphi(t) = \theta(0) + \int_{0}^{t} \kappa(\tau) v(\tau) \; d\tau
\end{equation*}

   The set of all values for $\theta$ has the topology of a circle whereas 
$\varphi$ can take on any real number value since it is the integral of the 
real valued function $\kappa(t) v(t)$.  We can recover $\theta$ from $\varphi$ 
by taking its value modulo $2\pi$.  The rate of change of $\theta$ and 
$\varphi$ are numerically equal, \ie $\dot{\varphi}(t) = \dot{\theta}(t)$.  By 
the Fox-Milnor theorem the total curvature of the curve from $t=t_1$ to $t=t_2$
is 
\begin{equation*}
\int_{t_1}^{t_2} \dot{\varphi}(\tau) \; d\tau = \varphi(t_2) - \varphi(t_1)
\end{equation*}
The {\it turning number} is the total curvature divided by $2\pi$.  It measures 
how far the tangent vector has turned over the length of the curve.

   The velocity of the curve can be expressed in terms of $\varphi(t)$ as
$v(t)\; (\cos(\varphi(t)), \; \sin(\varphi(t)))^T$.  Given the initial point of 
the curve, $(x(0),\; y(0))^T$, we can express the point at other times as
\begin{eqnarray} \label{E:eq1}
\begin{pmatrix}
x(t) \\ y(t) 
\end{pmatrix}
 &=& 
\begin{pmatrix}
x(0) \\ y(0) 
\end{pmatrix}
+
\int_0^t v(\tau) 
\begin{pmatrix}
\cos(\varphi(\tau)) \\ 
\sin(\varphi(\tau))
\end{pmatrix}
d\tau \nonumber \\
\end{eqnarray}

   We now assume that the speed and curvature are defined for all $t \in \real$
and that they are periodic functions with a common minimal period $T>0$.  This
can occur by $(x(t),\; y(t))^T$ having period $T$ but this is not necessary.  
We call an arc within such a curve whose domain is an interval of length $T$ a 
{\it periodic arc} of the curve.  We show how to partition such curves into 
congruent periodic arcs below.

   It follows that $\dot{\varphi}(t) = \kappa(t)v(t)$ is a periodic function 
with period $T$.  The integral of a periodic function is periodic if its 
average value over one period is zero.  And if we subtract the average value 
from a periodic function its integral will be periodic.  So we set
\begin{eqnarray*}
  \overline{\kappa} &=& \frac{1}{T} \int_0^{T} v(\tau) \kappa(\tau) \; d\tau \\
\widetilde{\varphi}(t) &=&  \kappa(0) v(0) + \int_0^t \kappa(\tau) v(\tau)
- \overline{\kappa} \;\; d\tau
\end{eqnarray*}
This allows us to write $\varphi(t) = \; \overline{\kappa}\; t + 
\widetilde{\varphi}(t)$ where $\widetilde{\varphi}(t)$ has period $T$.  We let
$R_{\phi}$ stands for a rotation by $\phi$ radians.  Even though $\varphi(t)$ 
is not periodic whenever $\overline{\kappa} \neq 0$ it is the case that 
\begin{lemma} \label{T:varphi}  For all $t \in \real$
\begin{eqnarray*} 
R_{\overline{\kappa} \; T}
\begin{pmatrix}
 \cos( \varphi(t) ) \\ 
 \sin( \varphi(t) )
\end{pmatrix}
=
\begin{pmatrix}
 \cos( \varphi(t+T) ) \\ 
 \sin( \varphi(t+T) )
\end{pmatrix}
\end{eqnarray*} 
\end{lemma}

\begin{proof}
After making the substitution $\varphi(t) = \overline{\kappa} t + 
\widetilde{\varphi}(t)$ the proof is just a calculation which makes use of 
matrix multiplication, addition rules from trigonometry, and the fact that 
$\widetilde{\varphi}(t)$ has period $T$.
\end{proof}
  
   The quantity $\bar{\kappa}\,T$ is the total curvature for the periodic arcs
of the curve.  We denote the turning number of a periodic arc by $\k$.  Note 
this is independent of the choice of periodic arc.  For any $t_0 \in \real$
\begin{eqnarray}\label{E:kappabreve}
\k &=& \frac{\bar{\kappa}\; T}{2\pi} \nonumber
   = \frac{1}{2\pi} \int_{t_0}^{t_0+T} \kappa(\tau) v(\tau) \; d\tau \qquad \\
~  &=& \frac{1}{2\pi} \int_0^{T} 
        \frac{ \dot{x}(\tau)\ddot{y}(\tau) - \dot{y}(\tau)\ddot{x}(\tau)}
                {\dot{x}(\tau)^2 + \dot{y}(\tau)^2 } \; d\tau \qquad
\end{eqnarray} 
An advantage of $\k$ is that, unlike arc length, it doesn't necessarily contain 
a radical under the integral which improves the prospects of finding an 
anti-derivative for use in the fundamental theorem of calculus.

   Winfree coined the term ``isogon contours'' in his study of spiral tip 
meander \cite{winfree91}.  We say here that an {\it isogonal curve} is a level 
curve of $\k$ whether $\k$ is seen as a function in the state space or as a 
function in the parameter space.

   An {\it even} congruence is a congruence of the Euclidean plane which 
preserves the orientation of the plane.  An {\it odd} congruence reverses the
orientation of the plane.  Total curvature is invariant under even congruences 
and turned into its negative by odd congruences.  Thus the quantity $|\k|$ is 
invariant under all congruences.  It gives us a geometric property of the 
curve.  In particular we can express some of the curve's symmetries in terms of 
$\k$.  Let 
\begin{small}
\begin{eqnarray}
\mathcal{G}_{\k, T} 
\begin{pmatrix}
x(t) \\ y(t) 
\end{pmatrix} 
= \qquad \qquad \qquad \qquad \qquad \qquad \qquad \nonumber \\
\left(
\begin{pmatrix}
x(T) \\ y(T) 
\end{pmatrix}
-
R_{2\pi\k}
\begin{pmatrix}
x(0) \\ y(0)
\end{pmatrix}
\right)  + R_{2\pi\k}
\begin{pmatrix}
x(t) \\ y(t)
\end{pmatrix} 
\quad
\end{eqnarray}
\end{small}
When $\breve{\kappa} \in {\bf Z}$ the rotation $R_{2\pi\k}$ reduces to the 
identity map and $\mathcal{G}_{\k, T}$ is a translation by the vector 
$(x(T)-x(0),\, y(T)-y(0))^T$.  Otherwise $\mathcal{G}_{\k, T}$ is a 
rotation by $2\pi\k$ modulo $2\pi$ radians about the point 
\begin{equation*}
\begin{pmatrix}
\bar{x} \\ \bar{y}
\end{pmatrix} 
=
\frac{1}{2 \sin(\pi \k)} R_{\pi(1/2-\k)}
\begin{pmatrix}
x(T) - x(0) \\ y(T) - y(0)
\end{pmatrix} 
\end{equation*}

\begin{thm}  For all $t \in \real$ 
\begin{eqnarray*}
\mathcal{G}_{\k, T}
\begin{pmatrix}
x(t) \\ y(t) 
\end{pmatrix}
=
\begin{pmatrix}
x(t+T) \\ y(t+T) 
\end{pmatrix}
\end{eqnarray*}
\end{thm}

\begin{proof}
Moving $(x(0),\; y(0))^T$ from the right hand side of \eqref{E:eq1} to the left 
hand side and applying the rotation $R_{2\pi\k}$ to both sides gives
\begin{small}
\begin{eqnarray*}
R_{2\pi\k}
\left(
\begin{pmatrix}
x(t) \\ y(t) 
\end{pmatrix}
- 
\begin{pmatrix}
x(0) \\ y(0) 
\end{pmatrix}
\right) \qquad \qquad \\
\qquad =
R_{2\pi\k}
\int_0^t v(\tau) 
\begin{pmatrix}
\cos(\varphi(\tau)) \\ 
\sin(\varphi(\tau))
\end{pmatrix}
d\tau \\
  =
\int_0^t v(\tau) 
R_{2\pi\k}
\begin{pmatrix}
\cos(\varphi(\tau)) \\ 
\sin(\varphi(\tau))
\end{pmatrix}
d\tau 
\end{eqnarray*}
\end{small}
By Lemma \ref{T:varphi} 
\begin{small}
\begin{eqnarray*}
\int_0^t v(\tau) 
R_{2\pi\k}
\begin{pmatrix}
\cos(\varphi(\tau)) \\ 
\sin(\varphi(\tau))
\end{pmatrix}
d\tau \qquad \qquad \\
  =
\int_0^t v(\tau) 
\begin{pmatrix}
\cos(\varphi(\tau+T)) \\ 
\sin(\varphi(\tau+T))
\end{pmatrix}
d\tau
\end{eqnarray*}
\end{small}
and since $v(t)$ has period $T$
\begin{small}
\begin{eqnarray*}
\int_0^t v(\tau) 
\begin{pmatrix}
\cos(\varphi(\tau+T)) \\ 
\sin(\varphi(\tau+T))
\end{pmatrix}
d\tau \qquad \qquad \\
=
\int_0^t v(\tau+T) 
\begin{pmatrix}
\cos(\varphi(\tau+T)) \\ 
\sin(\varphi(\tau+T))
\end{pmatrix}
d\tau
\end{eqnarray*}
\end{small}
Using the change of variables theorem with $\eta = \tau + T$ gives
\begin{small}
\begin{eqnarray*}
\int_0^t v(\tau+T) 
\begin{pmatrix}
\cos(\varphi(\tau+T)) \\ 
\sin(\varphi(\tau+T))
\end{pmatrix}
d\tau = \qquad \qquad \qquad \qquad  \\
\int_0^{t+T} v(\eta) 
\begin{pmatrix}
\cos(\varphi(\eta)) \\ 
\sin(\varphi(\eta))
\end{pmatrix}
d\eta
-
\int_0^{T} v(\eta) 
\begin{pmatrix}
\cos(\varphi(\eta)) \\ 
\sin(\varphi(\eta))
\end{pmatrix}
d\eta \\
\end{eqnarray*}
\end{small}
Therefore
\begin{small}
\begin{equation*}
R_{2\pi\k}
\begin{pmatrix}
x(t) \\ y(t) 
\end{pmatrix}
- R_{2\pi\k}
\begin{pmatrix}
x(0) \\ y(0) 
\end{pmatrix}
= 
\begin{pmatrix}
x(t+T) \\ y(t+T) 
\end{pmatrix}
-
\begin{pmatrix}
x(T) \\ y(T) 
\end{pmatrix}
\end{equation*}
\end{small}
which can be rearranged to give the theorem.
\end{proof}

\begin{figure}[ht]
\centering
\includegraphics[scale=0.55]{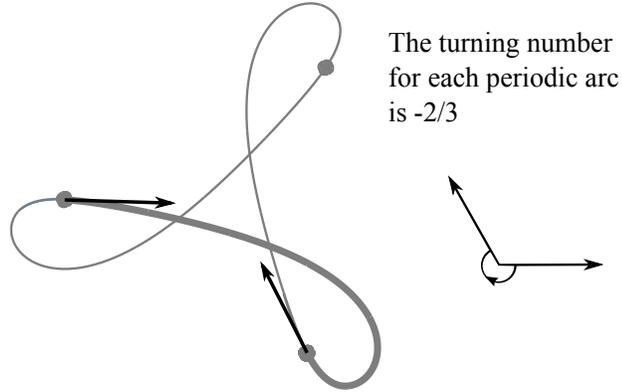}
\caption{A curve with $\k = -2/3$.  The curve is partitioned into three 
periodic arcs.  The tangent vector rotates by minus two thirds of a turn 
from the starting point of a periodic arc to its final point.  The tangent 
vector rotates by minus two whole turns along the entire closed curve so its
Whitney turning number is $-2$.  The curve has no reflectional symmetry and its 
full symmetry group is generated by $\mathcal{G}_{-2/3, T}$.}
\label{F:trefoil}
\end{figure}

   We can arbitrarily pick any point on the curve, $(x(t_1), \; y(t_1))^T$, and 
apply $\mathcal{G}_{\k, T}$ to it to get the point $(x(t_1+T), 
y(t_1+T))^T$.  The portion of the curve between these two points is a periodic 
arc.  We can apply $\mathcal{G}_{\k, T}$ to this periodic arc to get an 
adjacent periodic arc and so on.  We can apply the inverse of 
$\mathcal{G}_{\k,T}$ to get the rest of the curve on the other side of 
$(x(t_1), y(t_1))^T$.  In this way we can partition the curve into periodic 
arcs starting from any point.  A simple example is shown in figure 
\ref{F:trefoil}.

   The quantity $\k$ is like the Whitney turning number \cite{whitney} for 
closed curves except the Whitney turning number is a topological invariant 
whereas $\k$ is a geometric invariant.  Also the value of $\k$ can be any real 
number whereas the Whitney turning number must be an integer.  When $\k$ is a 
non-integral rational number $p/q$ ($p$, $q$ coprime) there is a close 
relationship between the two quantities.  A periodic arc will return to itself 
after being rotated $q$ times by $\mathcal{G}_{p/q, T}$.  This implies the 
image of the curve is closed and that the Whitney turning number of its image 
is $p$.  Conversely given a closed curve with Whitney turning number $p$ and 
rotational symmetry by $p/q$ of a turn the curve can be partitioned into $q$ 
arcs with total turning number $p/q$.  This is illustrated in figure 
\ref{F:trefoil}.

\section{Conservative examples}

   {\bf Example 1 - Epicyclic motion} \\ Epicyclic motion was an ancient Greek 
model for our solar system \footnote{Prominent figures in the development of 
this model were Apollonius, Hipparchus, and Ptolemy.  In Ptolemy's version the 
center was offset from the Earth but this did not change the shape of the curve 
in the Earth's rest frame.}.  It was a fair approximation for the motion of the 
planets as seen in the Earth's rest frame but of course it has long since been 
superseded by Heliocentric models. 

   Epicyclic motion is formed by combining two rotary motions.  A point on a
circle, called the {\it deferent}, revolves with a constant angular velocity
$\omega_1$.  At each moment in time the revolving point is the center for a 
another circle called the {\it epicycle} which spins about its center with 
constant angular velocity $\omega_2$.  A chosen point on the epicycle, that we 
will call the {\it tracing point}, stood for the location of a planet (see 
figure \ref{F:epicycle}).  The orbit of a planet was represented by the curve 
generated by the tracing point.  

\begin{figure}[ht]
\centering
\includegraphics[scale=0.42]{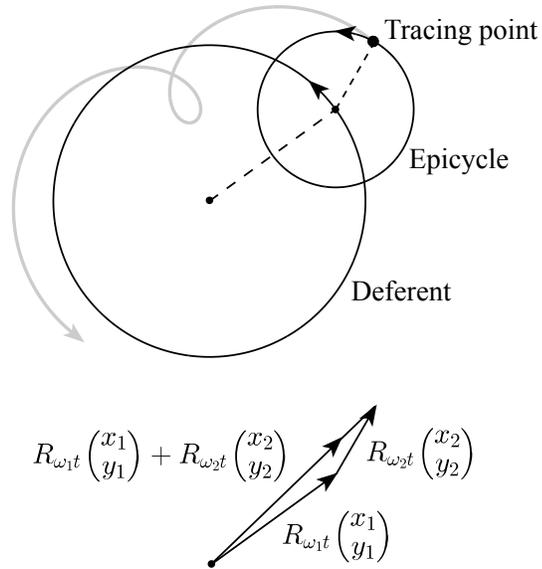}
\caption{Epicyclic motion is the spin of the epicycle as its center revolves
about the center of the deferent.  It can be expressed as the sum of two
vectors which turn with fixed angular velocities.}
\label{F:epicycle}
\end{figure}

   At any moment in time the location of the tracing point can be expressed as
the sum of two vectors: a position vector for the location of the center of the 
epicycle relative to the center of the deferent and a position vector for the 
location of the tracing point relative to the center of the epicycle.  We 
denote these two vectors at time $t=0$ as $(x_1,y_1)^T$ and $(x_2,y_2)^T$ 
respectively.  Figure \ref{F:epicycle} presents a simple time parametrization 
for curves traced out by epicyclic motion.

   Epicyclic motion can be generated by projecting the orbits for a pair of
uncoupled harmonic oscillators.  Harmonic oscillators are conservative systems.
In this case the dynamics can be obtained from the Hamiltonian
\begin{small}
\begin{equation*}
H(x_1,y_1,x_2,y_2) = - \frac{1}{2} \left(
 \omega_1 \left( x_1^2 + y_1^2 \right) + \omega_2 \left( x_2^2 + y_2^2 \right) 
                                                                       \right)
\end{equation*}
\end{small}
Because of the symmetry of this Hamiltonian it is arbitrary which pair of
variables, $x_1$, $x_2$ or $y_1$, $y_2$, are regarded as the positions
and which are regarded as the momenta.  The choice only affects the direction 
the system moves along the orbits.  We choose $x_1$, $x_2$ to be the position 
variables and $y_1$, $y_2$ to be conjugate momenta.  Each $(x_j,y_j)^T$ pair 
moves with angular frequency $\omega_j$ along a circle in the state space. 
   
   Projecting $(x_1,y_1,x_2,y_2)^T$ to the pair of position variables,
$(x_1,x_2)^T$, gives us Lissajous curves.  We will consider Lissajous curves 
only briefly.  This will occur in the following example on the spherical 
pendulum.  To obtain epicyclic motion from uncouple harmonic oscillators we 
instead project $(x_1,y_1,x_2,y_2)^T$ to $(x_1+y_1,\, x_2+y_2)^T$.

   We set $r_j = \left|\left| \, (x_j,y_j)^T \, \right|\right|$ for $j=1,2$.
For all $t \in \real$ we have 
\begin{equation*}
| r_1 - r_2 | \leq \left|\left|\, (x(t), y(t))^T \, \right|\right| \leq r_1+r_2
\end{equation*}
We call $| r_1 - r_2 |$ the {\it minimum radius} and $r_1+r_2$ the {\it maximum 
radius}.  The curve attains its \\
minimum radius when $R_{\omega_1 t} ((x_1, y_1)^T)$ and \\
$R_{\omega_2 t}((x_2, y_2)^T)$ point in opposite directions
and it attains its maximum radius when they point in the same direction.  By 
suitably shifting time we can suppose, without loss of generality, that at 
$t=0$ the vectors point in the same direction.  This simplifies the time 
parameterization of the curve to:
\begin{equation}\label{E:parameterization}
\begin{pmatrix}
x(t) \\ y(t)
\end{pmatrix}
=
R_{\omega_1 t} 
\begin{pmatrix}
r_1 \\ 0
\end{pmatrix}
+
R_{\omega_2 t} 
\begin{pmatrix}
r_2 \\ 0
\end{pmatrix}
\end{equation}

   If $\omega_1 = \omega_2$ then the vectors $R_{\omega_1 t}( r_1,0)^T$, \\
$R_{\omega_2 t}( r_2,0)^T$ will continue to point in the same direction and the
tracing point will travel in a circle.  Also if $\omega_1=0$ or $\omega_2=0$ 
the tracing point will travel in a circle so in this section we assume that
$\omega_1\omega_2(\omega_2-\omega_1) \neq 0$.  Furthermore, for the purpose of 
comparing these curves to the curves generate by spiral tip meander in section
\ref{S:dissipative}, we will assume $\omega_1>0$ which is to say the epicycle 
revolves anticlockwise.  

   Since $\omega_1 \neq \omega_2$ the vectors $R_{\omega_1 t}( r_1,0)^T$, \\
$R_{\omega_2 t}( r_2,0)^T$ will alternately point in the same and opposite 
directions.  The condition that they point in the same or opposite direction is 
equivalent to
\begin{small}
\begin{eqnarray*}
\left|\left| \, (\cos(\omega_1 t),\;\sin(\omega_1 t), \;0) \times
  (\cos(\omega_2 t),\;\sin(\omega_2 t), \;0) \, \right|\right| = \\
\sin((\omega_2-\omega_1)t)=0 \qquad \qquad \qquad \qquad
\end{eqnarray*}
\end{small}
Thus configurations for the deferent, epicycle, and tracing point which are
congruent to the initial condition occur with a periodicity of 
$T=2\pi/|\omega_2-\omega_1|$.  It can be checked that this is the common 
minimal period for $v(t)$ and $\kappa(t)$.  So we can apply the theory from the 
previous section to epicyclic motion.

   The curves generated by the tracing point are not properly called epicycles.
These curves have names based on a different construction method.  They can be 
constructed as roulettes, \ie by one curve rolling without slipping along 
another curve.  One of the simplest non-trivial roulettes is generated by a 
point on a disk rolling without slipping along a line.  These are called 
{\it trochoids}.  

   The curves generated by epicyclic motion can be generated by a circular disc 
rolling without slipping along another circular disc.  This construction method 
was perhaps originally conceived of by D\"urer in 1525 and then again by the 
astronomer R{\o}mer in 1624.  These curves have been studied by many 
mathematicians since.  

   For some purposes it is useful to allow the tracing point to be outside of 
the rolling disc.  We can treat the union of the rolling disc and the tracing 
point as a single rigid body even when their union does not form a connected 
set.  This is done by applying the same motion of the rolling disc to the 
tracing point regardless of where the tracing point happens to be.  And 
since we are allowing the tracing point to be outside of the rolling disc, we 
can dispense with the disc's interior in the definitions and just work with a
pair of circles, one fixed and one rolling along other.   

   The pair of circles are required to intersect in exactly one point.  When 
they intersect in exactly one point they have the same tangent line at the 
contact point, hence the circles are said to be tangent to each other.  If 
neither circle is inside the other then they are said to be externally tangent. 
Otherwise they are said to be internally tangent.  If the fixed and rolling
circles are externally tangent then the curve generated by the tracing point is
called an {\it epitrochoid}.  If the fixed circle is outside of the rolling
circle then the curve is called a {\it hypotrochoid}.  If the fixed circle is 
inside of the rolling circle then the curve is called a {\it peritrochoid}.

   If the tracing point is inside of the rolling circle then the hypotrochoid, 
epitrochoid, or peritrochoid is said to be {\it curtate}.  If the tracing 
point is outside of the rolling circle then the hypotrochoid, epitrochoid, or 
peritrochoid is said to be {\it prolate}\,\footnote{Some authors reverse the 
meaning of curtate and prolate, \eg \cite{ganguli}}.  If the tracing point is 
on the rolling circle then the hypotrochoid, epitrochoid, or peritrochoid is 
called a {\it hypocycloid}, {\it epicycloid}, or {\it pericycloid} respectively

\begin{figure}[ht]
\centering
\includegraphics[scale=0.275]{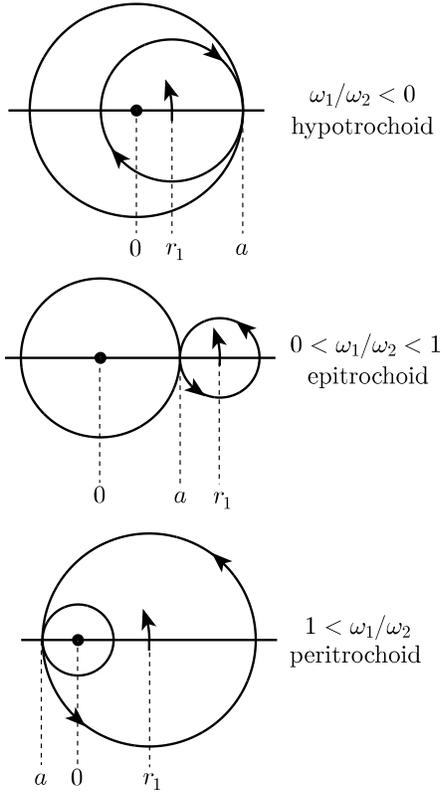}
\caption{Three configurations for the fixed and rolling circles.  We assume the 
rolling circles revolve anti-clockwise as indicated by the arrows at their 
centers.  This causes the rolling circle to turn clockwise for hypotrochoids 
and anti-clockwise for epitrochoids and peritrochoids as indicated by the 
arrows on the rolling circles.} 
\label{F:TrochoidTypes}
\end{figure}

   Its useful to have one term which encompasses hypotrochoids, epitrochoids, 
and peritrochoids.  Morely recognized this in 1894 and proposed using just 
``trochoids'' \cite{morely} even though this term is often restricted to 
the case where a circle rolls along a line.  More recently the term ``centered 
trochoid'' has been proposed \cite{ferreol}.  We shall use {\it central
trochoid} and use the term {\it central cycloid} to denote a hypocycloid, 
epicycloid, or pericycloid.  The central cycloids are those central trochoids 
that have cusps.

   Although the curves generated by epicyclic motion are central trochoids the 
deferent and epicycle are generally not the same as the fixed and rolling 
circles in the roulette construction.  However the fixed circle can be 
concentric with the deferent and at any moment in time the rolling circle can 
be concentric with the epicycle.  Only the radii may need to be different.

   To determine the correct radii for the fixed and rolling circles we make use
of the no slip condition.  The rolling circle may not slip as it goes around 
the fixed circle.  This means that the instantaneous velocity of the contact 
point is the zero vector $(0,0)^T$.  For epicyclic motion there is exactly one 
point, at any given moment in time, that is instantaneously at rest so this 
must be where the contact point for the fixed and rolling circles are at that 
moment.
   
   The contact point of any pair of tangent circles is collinear with their 
centers and by convention their centers are on the $x$-axis at $t=0$ so at this 
time the contact point is on the $x$-axis as well.  We let $a$ denote the 
$x$-coordinate of the contact point at time $t=0$.  So $|a|$ is the radius of 
the fixed circle.  

  At $t=0$ the vector $(r_1,0)^T$ is rotating with angular velocity $\omega_1$ 
about $(0,0)^T$ while the vector $(a,0)^T-(r_1,0)^T$ is rotating with angular
velocity $\omega_2$ about $(r_1,0)^T$.  Taking the derivative with respect to
time, evaluating at $t=0$, and setting the result equal to the zero vector 
gives
\begin{small}
\begin{equation*}
\begin{pmatrix}
0 \\ 0
\end{pmatrix}
=
\omega_1 R_{\pi/2}
\begin{pmatrix}
r_1 \\ 0
\end{pmatrix}
+
\omega_2 R_{\pi/2}
\left(
\begin{pmatrix}
a \\ 0
\end{pmatrix}
-
\begin{pmatrix}
r_1 \\ 0
\end{pmatrix}
\right)
\end{equation*}
\end{small}
which has the unique solution $a = (1-\omega_1/\omega_2)r_1$ (since $\omega_2 
\neq 0$).  To get the radius of the rolling wheel we set $b=a-r_1$,
\ie the directed distance from the center of the rolling circle to the contact 
point.  The radius of the rolling circle is $|b|$.

   The center of the fixed circle is at $(0,0)^T$ and at $t=0$ the center of 
the rolling circle is at $(r_1,0)^T$.  If $\omega_1/\omega_2<0$ then $r_1<a$, 
the contact point is on the right hand side of both circles, and the fixed 
circle is outside of the rolling circle.  So the curve is a hypotrochoid (see 
figure \ref{F:TrochoidTypes}).  If $0<\omega_1/\omega_2<1$ then $0<a<r_1$, the 
contact point is between the circles, and the circles are externally tangent.  
So the curve is an epitrochoid.  If $1< \omega_1/\omega_2$ then $a<0$, the 
contact point is on the left hand side of both circles, and the fixed circle is 
inside of the rolling circle.  So the curve is a peritrochoid.  

   From these facts the usual parameterizations for the central trochoids in 
terms of the radii $|a|$, $|b|$ and the angle $\phi=\omega_1 t$ can easily be 
derived.  It is convenient to work with \eqref{E:parameterization} because each
of the epicyclic parameters occurs just once in the expression.  It is helpful 
to keep in mind, though, the different ways that central trochoids can be 
constructed, \eg either by epicyclic motion or as a roulette.

\begin{figure}[ht]
\centering
\includegraphics[scale=0.5]{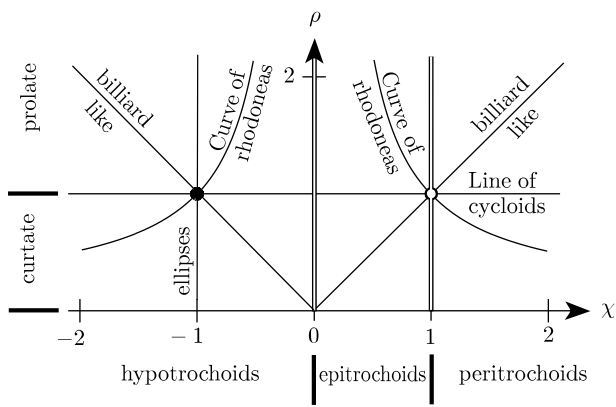}
\caption{The spaces of hypotrochoids, $\overline{\mathcal{H}}$, of 
epitrochoids, $\overline{\mathcal{E}}$, and of peritrochoids, 
$\overline{\mathcal{P}}$.  Along with the four distinguished cases, central
cycloids, ellipses, billiard like, and rhodoneas.}
\label{F:SpaceT}
\end{figure}

   We are primarily concerned with the shape of the central trochoids rather 
than their position, orientation, or size.  So we will determine the space of 
similarity classes of central trochoids.  For this purpose we introduce two 
geometrically invariant parameters for the central trochoids.

   One commonly used geometric invariant for describing the shape of a central 
trochoid is $ \rho = |\omega_2/\omega_1| r_2/r_1$.  This is often called the 
{\it arm ratio}.  This is because the line segment connecting the center of 
the rolling circle to the tracing point is often called the {\it arm} and 
$\rho$ is the ratio of the arm's length to the radius of the rolling circle, 
$r_2/|b|=\rho$.  When $\rho=0$ the resulting curve is just a single point if 
$\omega_1=0$ and a circle otherwise.  We do not wish to regard these as 
special cases of central trochoids so we require $\rho>0$.  When $0< \rho <1$ 
the central trochoid is curtate, when $\rho=1$ the central trochoid is a 
central cycloid, and when $\rho>1$ the central trochoid is prolate.

   Another geometrical invariant that we will use is the ratio of angular 
velocities $\chi = \omega_1/\omega_2$ (recall $\omega_2\neq 0$).  We call 
$\chi$ the {\it turning ratio}.  It can also be expressed as $\chi = 
(b/a)/(b/a-1)$.  The quantity $|b/a|$ is often called the {\it wheel ratio}.  
We will see that for curtate central trochoids the wheel ratio equals $\k$ (for
$\omega_1>0$).  

   Even with $\rho>0$ if $\chi=0,1$ then the curve that is traced out is a 
circle which we do not wish to include as a type of central trochoid.  When
$\chi<0$ the central trochoid is a hypocycloid, when $0<\chi<1$ the central
trochoid is an epitrochoid, and when $1<\chi$ the central trochoid is a
peritrochoid (see figure \ref{F:SpaceT}).  We set
\begin{eqnarray*}
\overline{\mathcal{H}} &=& \{\; (\chi,\rho)\; |\; \chi < 0,\; \rho > 0\; \} \\
\overline{\mathcal{E}} &=& \{\; (\chi,\rho)\; |\; 0 < \chi < 1,\; \rho>0\;\} \\
\overline{\mathcal{P}} &=& \{\; (\chi,\rho)\; |\; 1 < \chi,\; \rho > 0\; \}
\end{eqnarray*}
and $\overline{\mathcal{T}} = \overline{\mathcal{H}} \cup 
\overline{\mathcal{E}} \cup \overline{\mathcal{P}}$.

   In addition to the central cycloids there are three other cases of central 
trochoids worth distinguishing.  These cases also form curves in 
$\overline{\mathcal{T}}$ as shown in figure \ref{F:SpaceT}.  The first of these 
cases is given by the vertical line $\chi=-1$.  The equation for the 
hypotrochoid reduces to
\begin{equation*}
\begin{pmatrix}
x(t) \\ y(t) 
\end{pmatrix}
=
\begin{pmatrix}
(r_1+r_2) \cos(\omega_1 t) \\ (r_1-r_2) \sin(\omega_1 t) 
\end{pmatrix}
\end{equation*}
which determines an ellipse with semi-major axis $r_1+r_2$ and semi-minor axis
$|r_1-r_2|$.

   We call the cases given by the diagonal lines $\rho = |\chi|$ ``billiard 
like''.  These curves have long arcs with low curvature alternating with short 
arcs with high curvature.  They are fairly well approximated by the paths made 
by a frictionless billiard ball rolling on a circular table which travels along
straight line segments and bounce at the table's edge.

   The remaining distinguished case of central trochoids corresponds to the 
central trochoid passing through its own center of symmetry.  The minimum 
radius of a central trochoid is zero if and only if the deferent and epicycle 
have the same size.  When $r_1=r_2$ the distance of the tracing point from the 
center of the curve, as a function of time, is 
\begin{equation*}
                  2r_1\cos(((\omega_1-\omega_2)/2) \; t)
\end{equation*}
which is essentially the defining condition for ``rhodonea'' curves
\footnote{Studied by Guido Grandi around 1723}.  The condition $r_1=r_2$ 
is equivalent to $\rho\,|\chi| = 1$ which specifies a pair of hyperbolic arcs 
in $\overline{\mathcal{T}}$ (see figure \ref{F:SpaceT}).  These hyperbolic arcs 
form the subspace of rhodonea curves.  Below the hyperbolic arcs the deferent 
is larger than the epicycle while above the hyperbolic arcs the deferent is 
smaller than the epicycle.

\begin{table}
\centering
\renewcommand{\arraystretch}{1.8}
\begin{tabular}{|c||c|c|c|}
\hline
$\k$ & hypo- & epi- & peri- \\
\hline
\hline
prolate &  $\frac{1}{\chi-1}$ & $\frac{1}{1-\chi}$  & $\frac{1}{\chi-1}$ \\
\hline
curtate &  $\frac{\chi}{\chi-1}$ & $\frac{\chi}{1-\chi}$  & $\frac{\chi}{\chi-1}$ \\
\hline
\end{tabular}
\caption{The value of $\k$ for central trochoids when the rolling wheel 
revolves in the anticlockwise direction, $\omega_1>0$.  For $\omega_1<0$ take 
the negative of each table entry.}
\renewcommand{\arraystretch}{1}
\label{tbl:kformulas}
\end{table}

\begin{figure}[ht]
\centering
\includegraphics[scale=0.33]{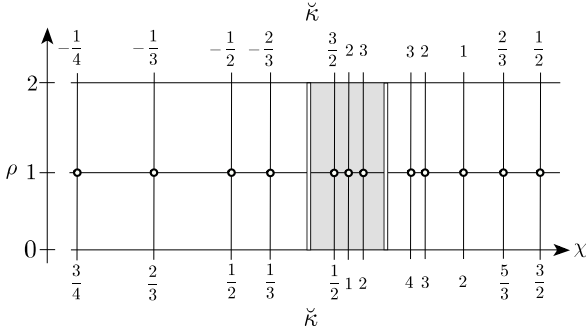}
\caption{The isogonal curves in $\overline{\mathcal{T}}$ are vertical line 
segments.  They are labeled with their value for $\k$.  The value of $\k$ 
changes by $\pm 1$ when the line of cycloids is crossed.  The region in 
$\overline{\mathcal{E}}$ is shaded to facilitate comparison to figure 
\ref{F:SpaceT}.}
\label{F:isogonal}
\end{figure}

   In $\overline{\mathcal{H}}$ the sets of hypocycloids, ellipses, billiard
like curves, and rhodoneas all intersect at one point $(\chi,\rho)=(-1,1)$.  
This particular special case is known as the ``Tusi couple''\footnote{Studied 
by Nasir al-Din al-Tusi around 1247}.  For the Tusi couple the tracing point 
goes back and forth along a line segment.  For the Tusi couple the tangent 
vector is undefined at the end points of the line segment, as with the cusps
of central cycloids.  The line segment can be regarded as a degenerate ellipse.
The line segment is literally a billiard curve for a round table so we can say 
it is a billiard like curve that passes through its center of symmetry.  

   The Tusi couple is the only instance in which the curvature of a central 
trochoid is zero at any point.  Hence central trochoids do not have inflection 
points.

   The value of $\k$ for central trochoids can be computed from equation 
\eqref{E:kappabreve}.  It almost reduces to an uncomplicated function of $\chi$ 
except that it depends on the signs of $\omega_1$, $\rho - 1$, and 
$\chi(\chi - 1)$.  The formulas for $\k$ are shown in table 
\ref{tbl:kformulas}.  For any real number except zero there is a central 
trochoid whose value for $\k$ is the given real number.

   Combining $\chi = (b/a)/(b/a-1)$ with $\k = \chi/(\chi-1)$ gives $\k = b/a$.
So we see from table \ref{tbl:kformulas} that $\k = b/a$ for curtate
hypotrochoids and peritrochoids.  For curtate epitrochoids $\k = -b/a$.  The
signs of $a$, $b$ are the same for hypotrochoids and peritrochoids and opposite
for curtate epitrochoids.  Therefore the total turning number for a periodic 
arc of a curtate central trochoid is the same as the wheel ratio, \ie $\k = 
|b/a|$.  

   The isogonal curves in $\overline{\mathcal{T}}$ are vertical line segments 
which span the heights of the curtate and prolate regions.  When the line of 
central cycloids is crossed the value of $\k$ changes by $\pm 1$ (see figure 
\ref{F:isogonal}).  The central trochoids undergo a cusp transition.  In this 
transition a loop is either added or removed from each periodic arc of the 
central trochoid.

   A sufficient, but not necessary, condition for two central trochoids to be 
similar is for them to have the same values for $\chi$ and $\rho$.  We will 
obtain the spaces of similarity classes of central trochoids by quotienting 
each of the spaces $\overline{\mathcal{H}}$, $\overline{\mathcal{E}}$, and 
$\overline{\mathcal{P}}$.

\begin{figure}[ht]
\centering
\includegraphics[scale=0.45]{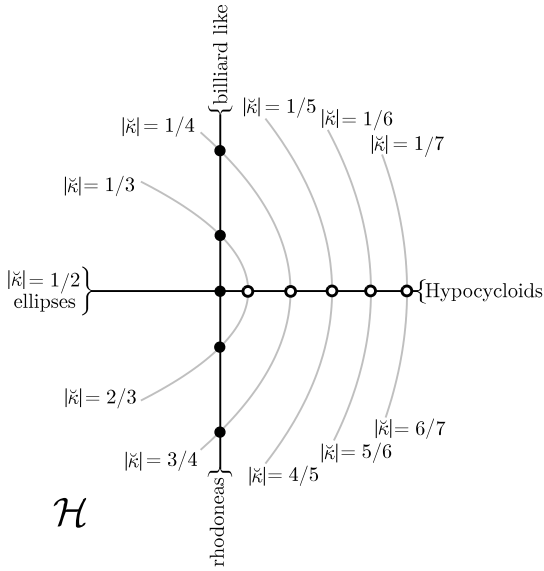}
\caption{The space, $\mathcal{H}$, of similarity classes of Hypotrochoids.  The
origin corresponds to the Tusi couple.  The four compass directions 
correspond to the four distinguished cases.  The isogonal curves are shown in 
gray and they are labeled with their $|\k|$ values.  The value of $|\k|$ jumps 
as an isogonal curve passes through the line of hypocycloids.}
\label{F:SpaceH}
\end{figure}

\begin{figure}[ht]
\centering
\includegraphics[scale=0.45]{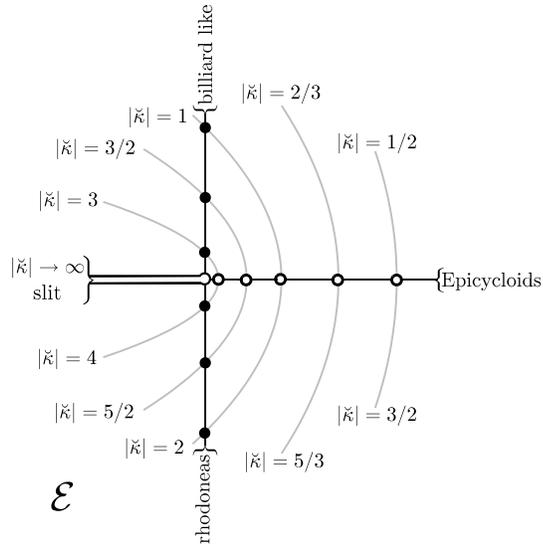}
\caption{The space, $\mathcal{E}$, of similarity classes of Epitrochoids.  The 
negative real axis and the origin make up the complement of $\mathcal{E}$ in
${\bf C}$.  The other three compass directions correspond to the three
distinguished cases of Epitrochoids.  The isogonal curves are shown in gray and 
they are labeled with their $|\k|$ values.  The value of $|\k|$ jumps as the 
isogonal curve passes through the line of epicycloids.}
\label{F:SpaceE}
\end{figure}

   In models for the planetary motion the epicycle was generally much smaller 
than the deferent.  Actually though, because vector addition is commutative, it 
doesn't matter which of the two circles we take to be the deferent and 
epicycle.  If we let the center of the deferent revolve about the center of the 
epicycle with angular velocity $\omega_2$ while the deferent spins with angular 
velocity $\omega_1$ then the exact same curve can be traced out.  However when 
we look for fixed and rolling circles to generate the central trochoid we get 
a different location for the point at instantaneous rest and thus different 
radii for the fixed and rolling circles.  This fact is known as David 
Bernoulli's ``dual generation theorem''.  
   
   A consequence of the dual generation theorem is that the ordered pair 
$(1/\chi,1/\rho)$ determines the same similarity class of central trochoids as 
$(\chi,\rho)$.  A curtate hypotrochoid is geometrically similar to a prolate 
hypotrochoid and {\it visa verse}.  A prolate epitrochoid is similar to a 
curtate peritrochoid and {\it visa versa}.  A curtate epitrochoid is similar to 
a prolate peritrochoid and {\it visa versa}.  

   The terms ``hypotrochoid'', ``epitrochoid'', ``peritrochoid'',
``curtate'', and ``prolate'' do not designate geometric similarity classes of
curves but merely describe how the curves can be constructed geometrically.
Many geometric figures can be constructed in more than one way.  This
terminology can lead to some confusion particularly since peritrochoids have
been regarded as epitrochoids by some authors while other authors have regarded
them as hypotrochoids.  Willson gives an overview on this terminology in the
appendix to his 1898 book \cite{willson}.

   Every central trochoid, except for central cycloids, can be constructed as a 
curtate central trochoid and as a prolate central trochoid.  The value of 
$|\k|$ equals the wheel ratio in the curtate method of construction.  We can
think of $\k$ as a generalization of the wheel ratio to other types of curves
with periodically varying curvature.  Although $\k$ can equal $0$ for some
curves with periodically varying curvature but not for central trochoids.

   Suppose we have two central trochoids with parameters $\omega_1$, 
$\omega_2$, $r_1$, $r_2$ and $\omega_1'$, $\omega_2'$, $r_1'$, $r_2'$.  A 
necessary condition for two central trochoids to be similar is for the ratio of 
their minimum radius to their maximum radius be the same, \ie
\begin{equation*}
\frac{|r_1-r_2|}{r_1+r_2} = \frac{|r_1'-r_2'|}{r_1'+r_2'}
\end{equation*}
Since the radii are positive this equation is equivalent to the equation
\begin{equation*}
\left(\frac{r_1'}{r_2'} - \frac{r_1}{r_2}\right)
\left(\frac{r_1'}{r_2'} - \frac{r_2}{r_1}\right) =0
\end{equation*}
If $r_1'/r_2'=r_1/r_2$ then its necessary for $\omega_1'/\omega_2'= 
\omega_1/\omega_2$ and if $r_1'/r_2'=r_2/r_1$ then its necessary for 
$\omega_1'/\omega_2'= \omega_2/\omega_1$.  Therefore the only points in 
$\overline{\mathcal{T}}$ that correspond to the same similarity class as
$(\chi,\rho)$ is $(1/\chi,1/\rho)$.

   We denote the space of similarity classes of hypotrochoids by 
${\mathcal{H}}$.  To denote the members of ${\mathcal{H}}$ we use the term 
{\it Hypotrochoid} with the first letter capitalized.  The word
``hypotrochoid'' with all lower case letters describes how the curve was
constructed.  For $(\chi,\rho) \in {\overline{\mathcal{H}}}$ the function 
\begin{equation*}
(\chi,\rho) \mapsto (\log(-\chi) + i \log(\rho))^2
\end{equation*}
is onto the complex plane, ${\bf C}$, and it is two to one everywhere except at
the Tusi couple.  The point with the same image as $(\chi,\rho)$ is 
$(1/\chi,1/\rho)$ so we can identify ${\mathcal{H}}$ with ${\bf C}$ (see 
figure \ref{F:SpaceH}).

   The value of $\chi$ is negative for hypotrochoids so regardless of whether 
it is prolate or curtate $-1< \k < 1$.  The projection from 
$\overline{\mathcal{H}}$ to ${\mathcal{H}}$ maps the isogonal curves with
$\k = \pm 1/2$ to the line of ellipses.  Otherwise it maps pairs of isogonal 
curves with opposite sign to the same semiparabolic arc in ${\mathcal{H}}$.  We 
can associate the value of $|\k|$ to each semiparabolic arc.  The semiparabolic
arcs whose values for $|\k|$ sum to $1$ form the entire parabola except for its 
vertex on the line of Hypocycloids.

   We denote the space of similarity classes of epitrochoids by 
${\mathcal{E}}$.  For members of ${\mathcal{E}}$ we use the term 
{\it Epitrochoid} with the first letter capitalized.  The words ``epitrochoid'' 
and ``peritrochoid'' with all lower case letters describes how the curve was 
constructed.  For $(\chi,\rho) \in {\overline{\mathcal{E}}} \cup 
{\overline{\mathcal{P}}}$ the function 
\begin{equation*}
(\chi,\rho) \mapsto (\log(\chi) + i \log(\rho))^2
\end{equation*}
is onto the complex plane ${\bf C}$ minus the non-positive real axis.  It is 
two to one everywhere in ${\overline{\mathcal{E}}} \cup 
{\overline{\mathcal{P}}}$.  The point with the same image as $(\chi,\rho)$ is 
$(1/\chi,1/\rho)$ so we can identify ${\mathcal{E}}$ with the slitted complex 
plane (see figure \ref{F:SpaceE}).

   The value of $\chi$ is positive for an epitrochoid or peritrochoid (for 
$\omega_1>0$) so regardless of whether it is prolate or curtate $\k > 0$.  The
projection from $\overline{\mathcal{E}}$ to ${\mathcal{E}}$ maps pairs of 
isogonal curves with the same $\k$ value to the same semiparabolic arc in 
${\mathcal{E}}$.  To obviate the issue of how the Epitrochoids are 
parameterized we associate the value of $|\k|$ to each semiparabolic arc.  If 
the difference in $|\k|$ between a semiparabolic arc in the lower half-plane 
of $\mathcal{E}$ and a semiparabolic arc in the upper half-plane of 
$\mathcal{E}$ is $1$ then their union is the entire parabola except for its 
vertex on the line of Epicycloids.

\smallskip

\noindent
{\bf Example 2 - The spherical pendulum} \\ The spherical pendulum is an
idealized mechanical system.  It is comprised of a weightless inextensible {\it
rod}, essentially a line segment with length $\ell$.  One end of the rod, the
{\it pivot}, is motionless for all time.  The other end, the {\it bob}, is the 
location of a point particle with mass $m$.  The bob is constrained to move 
on a sphere of radius $\ell$ centered at the pivot while subjected to a uniform 
gravitational field with strength $g$.

   We take the pivot to be the origin of a Cartesian coordinate system for the
spherical pendulum.  We take the direction of the gravitational field to be the
negative direction of the $z$-axis.  The $z$-axis is also called the {\it
pendulum's axis}.  The plane through the pivot and orthogonal to the pendulum's
axis is the {\it support plane}.  The support plane contains the $(x,y)$-axes
of the coordinate system (see figure \ref{F:spherical}).  Because the spherical
pendulum is symmetrical about its axis the orientation of the $(x,y)$-axes is 
completely arbitrary.

\begin{figure}[ht]
\centering
\includegraphics[scale=1.0]{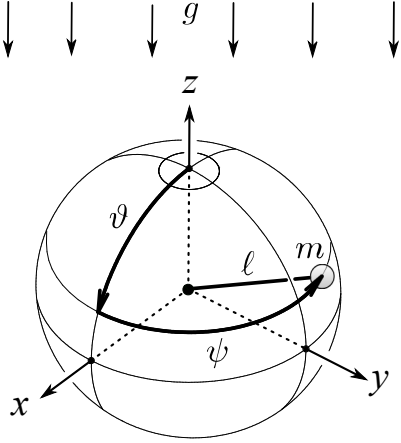}
\caption{Cartesian coordinates, $(x,y,z)$, and spherical coordinates, 
$(\vartheta, \psi)$ for the bob of a spherical pendulum. They are related by
\qquad \qquad \qquad \qquad \qquad \qquad \qquad \qquad $(x,y,z) = \ell \; 
(\sin(\vartheta)\cos(\psi),  \, \sin(\vartheta)\sin(\psi), \, 
\cos(\vartheta) )$.}
\label{F:spherical}
\end{figure}

   We also use spherical coordinates to specify the position of the bob.  
Because the letters $\theta$, $\varphi$ are used through out this article to 
describe the direction of the velocity along planar curves we let $\vartheta$ 
stand for the polar angle of the spherical coordinate system and $\psi$ stand 
for the azimuthal angle (see figure \ref{F:spherical}).

   The theory presented in section \ref{S:kappabreve} can be generalized to
curves on a sphere using the concept of geodesic curvature but to avoid
unnecessary complications in this example we will study the path of the bob 
from a bird's eye view.  More precisely stated we orthogonally project the 
bob's path on the sphere into the support plane.  We treat the spherical 
pendulum as merely a mechanism for generating planar curves which are the 
objects of our study here.

   We present a short review of the spherical pendulum's dynamics.  The bob 
moves in three dimensional space but it is subject to a single holonomic 
constraint so it has two degrees of freedom.  Two independent integrals of 
motion are the total energy or Hamiltonian, $H$ and the vertical component of 
its angular momentum, $J$.  The system is fully integrable.  Its orbits are 
either quasiperiodic, periodic, or fixed points.  

   The projection of the bob's position to the pendulum's axis is always 
periodic.  When the bob's $z$-coordinate has a minimal period it is sometimes 
referred to as the {\it pendulum's period} even if the pendulum is behaving 
quasiperiodically.  So long as $J \neq 0$ the bob's $(x,y)$ coordinates will be 
rotated about the pendulum's axis by a nonzero amount during the pendulum's 
period.  Typically the angle for this rotation is irrational so the overall 
motion and its projection to the support plane is quasiperiodic.  It will be 
shown that the pendulum's period is the common minimal period of $v(t)$, 
$\kappa(t)$ for the $(x,y)$-curves so we can apply the theory from section 
\ref{S:kappabreve}.

   The kinetic energy of the bob is
\begin{equation*}
\frac{1}{2} m (\dot{x}^2 + \dot{y}^2 + \dot{z}^2) = 
\frac{1}{2} m \ell^2 (\dot{\vartheta}^2 + \sin^2(\vartheta) \dot{\psi}^2 ) 
\end{equation*}
and the potential energy is $U = mgz = mg\ell\cos(\vartheta)$.  Since $U$
is independent of $\dot{\vartheta}$ and $\dot{\psi}$ the conjugate momenta are
\begin{small}
\begin{eqnarray*}
P_{\vartheta} &=& \frac{\partial}{\partial \dot\vartheta} \; \frac{1}{2} m 
\ell^2 
(\dot{\vartheta}^2 + \sin^2(\vartheta) \dot{\psi}^2 ) = m \ell^2 
\dot{\vartheta} \\
J &=& \frac{\partial}{\partial \dot\psi} \; \frac{1}{2} m \ell^2 
(\dot{\vartheta}^2 + \sin^2(\vartheta) \dot{\psi}^2 ) = m \ell^2 
\sin^2(\vartheta) \dot{\psi} 
\end{eqnarray*}
\end{small}
These are the horizontal and vertical components of the bob's total angular 
momentum respectively.  The state of the spherical pendulum is completely 
specified by the canonical variables $(\vartheta,P_{\vartheta},\psi,J)$.  The
Hamiltonian is
\begin{small}
\begin{equation} \label{E:hamiltonian}
H(\vartheta,P_{\vartheta},\psi,J) 
= \frac{1}{2 m \ell^2} \left( 
P_{\vartheta}^2 + \frac{J^2}{\sin^2(\vartheta)}\right) + m g \ell 
\cos(\vartheta)
\end{equation}
\end{small}

   The fact that $H$ is independent of $\psi$ shows us that the value of $J$ is
constant.  Although the state space is four dimensional the dynamics can be
reduced to two dimensions because $H$, $J$ are constant.  Moreover the reduced
system has a nondimensionalized form.  Physically, the parameters, $m$, $g$, 
$\ell$, are limited to positive values and varying them does not produce any 
qualitative changes in behavior so long as they remain positive.  To obtain the 
reduced system we nondimensionalized the constants of motion, we define a 
dimensionless potential energy and its dimensionless rate of change, and we
define a dimensionless time:
\begin{eqnarray*}
(h,j)    &=& \left(\frac{H}{mg\ell}, \; \frac{J}{m\ell\sqrt{g\ell}}\right) \\
(u, w)^T &=& \left(\frac{U}{mg\ell}, \; \sqrt{\frac{\ell}{g}} \; 
             \dot{u}\right)^T \\
\mathfrak{t} &=& \sqrt{g/\ell} \; t
\end{eqnarray*}
The equations of motion which can be obtained from the Hamiltonian 
\eqref{E:hamiltonian} can be used to show that $(u,w)^T$ satisfies the 
differential equation
\begin{equation} \label{E:reducedeq}
\begin{pmatrix}
d u/d \mathfrak{t} \\
d w/d \mathfrak{t} 
\end{pmatrix}
=
\begin{pmatrix}
w \\
3 u^2 - 2 h u - 1 
\end{pmatrix}
\end{equation}
This is known as the {\it reduced system} for the spherical pendulum.  Its is 
not a straight forward initial value problem.  First of all the initial value 
for $u$ must be in the interval $[-1,1]$.  Secondly it follows from 
\eqref{E:hamiltonian} that 
\begin{equation} \label{E:weq}
h = \frac{1}{2} \frac{w^2 + j^2}{1-u^2} + u
\end{equation}
Although $(u, w)^T$ varies with time the value of $h$ depends on $(u, w)^T$ in 
such a way that it does not change with time.  The value of $h$ can be 
determined by \eqref{E:weq} from the initial value for $(u, w)^T$ and the 
constant value for $j$.  Once the value of $h$ has been determined from the
initial conditions it can be treated as a fixed parameter in 
\eqref{E:reducedeq}.

   The $(x,y)$-curves can be obtained from a solution for $u$ by using a single 
quadrature.
\begin{eqnarray}\label{E:xycurves}
\psi(t) &=& \psi(0) + j \int_0^{\sqrt{g/\ell}\; t} \frac{d\tau}{1-u(\tau)^2} 
\nonumber \\
\begin{pmatrix}
x(t) \\ y(t) 
\end{pmatrix}
&=& \ell \sqrt{1-u(t)^2}
\begin{pmatrix}
\cos(\psi(t)) \\ \sin(\psi(t))
\end{pmatrix}
\end{eqnarray}

\begin{figure}[ht]
\centering
\includegraphics[scale=0.370]{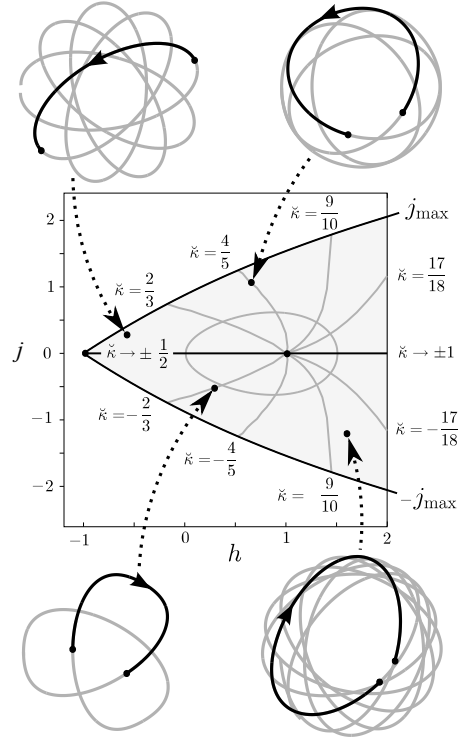}
\caption{The energy-momentum space, $\mathscr{P}$, for the spherical pendulum
is the light gray region bounded by the black $\pm j_{\max}$ curves.  The 
isogonal curves, shown in dark gray, are labeled with their values for $\k$.  
They radiate from $(1,0)$.  The dark gray oval, $27 j^2 = 2h(9-4h^2)$, is where 
inflection points partition the corresponding $(x,y)$-curves into periodic 
arcs.  Outside of the oval the $(x,y)$-curves do not have inflection points.
Insets show $(x,y)$-curves in gray each with a periodic arc shown in black.  
Dotted arrows point to the corresponding $(h,j)$ values.}
\label{F:hjspace}
\end{figure}

   The maximum potential energy is attained when the bob is directly above the 
pivot and the minimum potential energy is attained when the bob is directly 
below the pivot.  If the bob has no kinetic energy when its directly above or
below the pivot then it will remain where it is.  These two states are the 
fixed points of the Hamiltonian system \eqref{E:hamiltonian}.  If $h = -1$ then 
the bob must be motionless directly below the pivot.  This is the minimum 
possible value for $h$.  There is no limit to how fast the bob can move so $h$ 
has no upper bound.

   For each $h \geq -1$ there is a finite range of values for $j$.  The extreme 
values for $j$ can be obtained by rearranging \eqref{E:weq} to
\begin{equation}\label{E:lh}
-\frac{1}{2}\; j^2 = \frac{1}{2} \; w^2 - (h-u)(1-u^2) 
\end{equation}
and differentiating the right hand side with respect to $u$ and $w$ while 
treating $h$ as a fixed parameter.  There is one critical point for $j$ which
is $((h -\sqrt{h^2+3})/3,\; 0)^T$.  The maximum value for $j$ is
\begin{equation*}
j_{\max} = \frac{2}{9} \sqrt{ 3\; ((h^2+3)^{3/2} - h^3 + 9h) }
\end{equation*}
and the minimum value is $-j_{\max}$.  The graph of $j_{\max}$ is an 
increasing, concave down, curve with a single end point as shown in figure 
\ref{F:hjspace}.  The set of all possible values for $(h,j)$ is the {\it 
energy-momentum} space for the spherical pendulum,
\begin{equation*}
\mathscr{P} = \{\, (h,j) \, : \,h \geq -1, \, |j| \leq j_{\max} \}
\end{equation*}

   For $j = 0$ the angular velocity $\dot{\psi}$ is always zero and the 
pendulum moves within a vertical plane.  In this case it is often called a {\it
planar pendulum} even though there are no physical forces constraining it 
within a plane.  If $(h,j) = (-1,0)$ the $(x,y)$-curve is just a point.  If 
$(h,j) = (1,0)$ the bob can be directly above the pivot or it can move 
asymptotically towards that position.  For each $h$ such that $h > -1$ and $h 
\neq 1$ all of the $(x,y)$-curves corresponding to $(h,0)$ are line segments 
with the same length.

    The right hand side of \eqref{E:lh} can be taken as a Hamiltonian function 
for the reduced system \eqref{E:reducedeq}.  Thus the critical point of $j$ 
corresponds to a fixed point of the reduced system.  So the extreme values 
for $j$ are attained when $w=0$, \ie the potential energy is constant.  In 
these cases the height of the bob does not change.  This is often referred 
to as a {\it conical pendulum} because the rod sweeps out a cone.  The bob 
rotates along a horizontal circle with the sign of $j$ determining the 
direction of rotation.  For each $h > -1$ all of the $(x,y)$-curves determined 
by $(h,\; \pm j_{\max})$ are the same circle.

\begin{figure}[ht]
\centering
\includegraphics[scale=0.7]{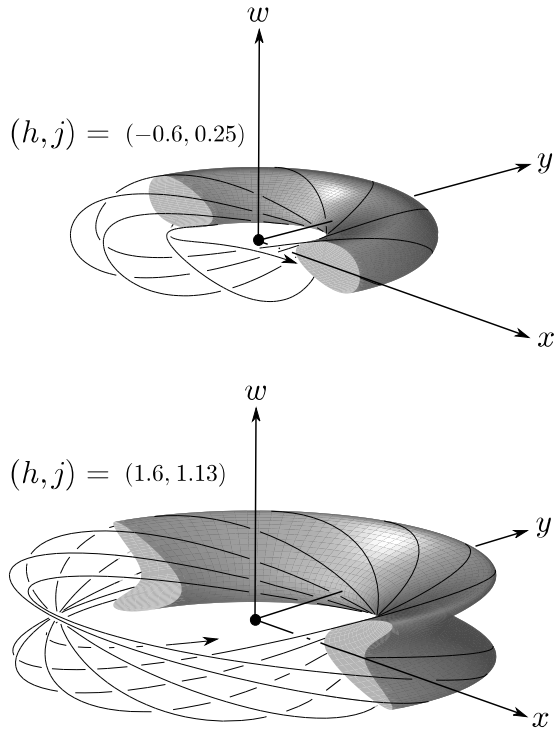}
\caption{The projection of the invariant tori corresponding to given $(h,j)$ 
values into the coordinate space for $(x,y,w)$.  The location of the $(h,j)$
values in $\mathscr{P}$ is shown in figure \ref{F:hjspace}.  Each torus shows 
the projection of a rotationally symmetric quasiperiodic orbit in the torus.}
\label{F:tori}
\end{figure}

   For each $(h,j)$ in 
\begin{equation*}
\mathscr{R} = \{\, (h,j) \, : \,h > -1, \, 0 < |j| < j_{\max} \}
\end{equation*}
the subset of the spherical pendulum's state space which is mapped to $(h,j)$
is a torus\footnote{Cushman's $\mathscr{R}$ space is a little larger than 
this.} \cite{cushman}.  The action of the group of rotations about the 
pendulum's axis can be extended to the whole state space of the pendulum and 
for each $(h,j) \in \mathscr{R}$ the corresponding torus is invariant under 
this action.  Rotations about the pendulum's axis have no effect on $w$ and 
these tori can be projected into a three dimensional space while preserving 
their symmetry by using the coordinates $(x,y,w)$ (see figure \ref{F:tori}).  
For each $(h,j) \in \mathscr{R}$ all of the corresponding $(x,y)$-curves are 
rotated copies of each other.  The orbits typically wind quasiperiodically 
around the torus but they can be periodic.  In either case the curvature of the 
$(x,y)$-curve varies periodically.  The value of $\k$ is defined for the 
corresponding $(x,y)$-curve if and only if $(h,j) \in \mathscr{R}$.

   Since the shape of the $(x,y)$-curves is completely determined by $(h,j) 
\in \mathscr{R}$ we can think of $\k$ as a function on the space $\mathscr{R}$.
We let $T$ denote the period of the reduced system in dimensionless time 
$\mathfrak{t}$.  $T$ depends on $(h,j)$.  In physical time $t$ the period is 
$\sqrt{\ell/g} \; T$.  We obtain the velocity of the $(x,y)$-curves from 
\eqref{E:xycurves} and 
\eqref{E:lh}.
\begin{equation*}
v = \sqrt{g\ell \, (2h -2u - w^2)}
\end{equation*}
The period of $2h -2u - w^2$ is the same as the period of the reduced system 
and since the expression under the radical is positive for all time for $(h,j)
\in \mathscr{R}$ the period of $v$ is the same as the period of the reduced 
system.  We also obtain the curvature from \eqref{E:xycurves} and \eqref{E:lh}.
\begin{equation*}
\kappa = \frac{\dot{x} \ddot{y} - \dot{y}\ddot{x}}{v^3} = j \sqrt{g^3\ell} 
\; \frac{2h - 3u}{v^3}
\end{equation*}
which has the same period.  From equation \eqref{E:kappabreve} and the change 
of variables theorem we get an expression for $\k$ entirely in terms of the 
reduced system.
\begin{equation}
\k = \frac{j}{2\pi} \int_0^T \frac{2h -3u}{2h - 2u -w^2} \; d\tau 
\end{equation}
A few numerically computed isogonal curves in $\mathscr{R}$ are shown in 
figure \ref{F:hjspace}.  The isogonal curves radiate from the point 
$(h,j)=(1,0)$.  

   It is common to approximate the behavior of the spherical pendulum when
$(h,j)$ is near the vertex $(-1,0) \in \mathscr{P}$ by linearizing the system
about its fixed point for $(h,j)=(-1,0)$.  In this case we might expect that 
the $(x,y)$-curves for the spherical pendulum could be well approximated by
Lissajous curves since we are projecting the state variables to the position
variables $(x,y)$.  However because the spherical pendulum is symmetrical about 
its axis the two frequencies of the linearized system are equal and therefore, 
regardless of the initial conditions, the only type of Lissajous curves 
generated by projecting the state variables of the linearized system to the 
position variables are ellipses or line segments.  Furthermore the eigenvalues 
of the linearized system are purely imaginary so the Hartman-Grobman theorem 
does not apply \cite{hartman}, \ie there need not be a neighborhood in which 
the spherical pendulum is equivalent to its linearization.  The spherical 
pendulum is in fact highly nonlinear \cite{malkin}.  No matter how close $(h,j)
\in \mathscr{R}$ is to $(-1,0)$ the $(x,y)$-curves resemble Hypotrochoids more
than they do Lissajous curves (see insets in figure \ref{F:hjspace}), which is
not too surprising given the geometry of the spherical pendulum.

    It is interesting to compare the spaces $\mathscr{P}$ and $\mathscr{R}$ 
(figure \ref{F:hjspace}) with the spaces $\overline{\mathcal{H}}$ (figure 
\ref{F:isogonal}) and $\mathcal{H}$ (figure \ref{F:SpaceH}).  For these spaces 
the reflection about the horizontal axis maps isogonal curves to isogonal 
curves and the range of observed values for $\k$ is in the open interval 
$(-1,1)$.  Also in these spaces the range of observed values for $|\k|$ below 
the horizontal axis is the open interval $(1/2,1)$.

   Numerical analysis indicates that in $\mathscr{R}$ the value of $|\k|$ can 
be within any tiny distance above $1/2$ but it can not be $1/2$ or less.  For 
$h \in (-1,1)$ the value of $\k$ appears to converge to $1/2$ as $j$ approaches 
$0$ from above while it appears to converge to $-1/2$ as $j$ approaches $0$ 
from below.  Numerical analysis also indicates that the value of $|\k|$ can be 
within any tiny distance below $1$ but it can not be $1$ or more.  For $h>1$ 
the value of $\k$ appears to converge to $1$ as $j$ approaches $0$ from above 
while it appears to converge to $-1$ as $j$ approaches $0$ from below.  As 
$(h,j)$ crosses the horizontal axis of $\mathscr{P}$ the $(x,y)$-curves 
transition by collapsing to a line segment.

   In $\overline{\mathcal{H}}$ the value of $\k$ converges to $1/2$ as 
$(\chi, \rho)$ approaches the Tusi couple, $(-1,1)$, from below and it 
converges to $-1/2$ as $(\chi, \rho)$ approaches the Tusi couple from above.  
As $(\chi, \rho)$ crosses the Tusi couple the $(x,y)$-curves transition by 
collapsing to a line segment.

   In $\mathscr{P}$ there is a half-line for which the corresponding 
$(x,y)$-curves are line segments while in $\overline{\mathcal{H}}$ there is a
single point for which the $(x,y)$-curves is a line segment.  On the other hand
in $\mathscr{P}$ the isogonal curves radiant from a single point while in
$\mathcal{H}$ all of the isogonal curves pass through a half-line.

   There are some important differences between in $\mathscr{P}$ and 
$\mathcal{H}$.  The value of $\k$ for billiard like Hypotrochoids must be in 
the interval $[-1/2,1/2]$ so the $(x,y)$-curves generated by the spherical 
pendulum do not resemble billiard like Hypotrochoids.  They also do not 
resemble rhodonea curves since they only pass through their center of symmetry 
when they collapse to line segments.  And unlike Hypotrochoids an $(x,y)$-curve 
for the spherical pendulum can have inflection points.  This happens when 
$(h,j)$ is inside the oval shown in figure \ref{F:hjspace}.

   The value of $\k$ is associated with the intrinsic precession of the 
spherical pendulum.  The motion of the projected image of the bob in the 
support plane can be thought of as a compound motion of a point around an 
ellipse with the rotation of the ellipse about its center.  The {\it intrinsic 
precession} of the spherical pendulum is the rotation of the ellipse.

   It should be briefly pointed out that the intrinsic precession of a
spherical pendulum is distinct from Foucault precession.  The difference was
recognized by Foucault himself.  As is well known, Foucault designed and built 
a spherical pendulum in 1851 to measure the rotation of the Earth 
\cite{foucault}.  Foucault's pendulum was designed to be set librating within a
vertical plane.  Since the pendulum's support is rotating with the Earth the
plane of libration appears to rotate relative to the ground below it.  This
motion is called {\it Foucault precession}.

   Foucault found that it can be difficult to start a spherical pendulum with 
sufficiently little angular momentum so that it will appear to oscillate within
a vertical plane.  Even a small amount of angular momentum led to an intrinsic 
precession comparable to the Foucault precession.  To overcome this apparent 
``instability'' he designed his pendulum with very large $m$ and $\ell$.  

   If $|j|$ is small enough the periodic arcs of the $(x,y)$-curve generated by
the spherical pendulum can be fairly well approximated by the periodic arcs of
an ellipse (see top left inset in figure \ref{F:hjspace}).  The approximating
ellipse turns in the same direction as the bob so the total curvature of a
periodic arc of the $(x,y)$-curve of the spherical pendulum is the sum of the
total curvature of a periodic arc of the approximating ellipse with the
intrinsic precession of the approximating ellipse.  For any nondegenerate 
ellipse the value of the total curvature of a periodic arc is $\k = \pm 1/2$ 
(recall figure \ref{F:SpaceH}).  The amount of intrinsic precession that 
occurs during the pendulum's period is
\begin{equation*}
\left\{
\begin{array}{cl}
 \k - 1/2  &  \mathrm{for} \quad \k >  1/2 \\
 \k + 1/2  &  \mathrm{for} \quad \k < -1/2 
\end{array}
\right.
\end{equation*}
If $|j|$ is large the periodic arcs of the $(x,y)$-curve are not well 
approximated by periodic arcs of an ellipse (see bottom left inset in figure 
\ref{F:hjspace}).  In these cases it is not very helpful to think of the 
$(x,y)$-curve as being generated by a precessing ellipse.  

   The value of $\k$ is defined for all $(h,j) \in \mathscr{R}$ regardless of 
how poorly the periodic arcs of an $(x,y)$-curve can be approximated by the 
periodic arcs of an ellipse.  The value of $\k$ tells us how far the 
$(x,y)$-curve turns during the pendulum's period as well as providing us with
information about the symmetry of the $(x,y)$-curve.

\section{Dissipative examples}
\label{S:dissipative}

   In this section we consider two related models for natural systems that 
generate curves with periodic curvature.  These are the Shenoy-Rutenberg model
for the paths taken by the bacterium {\it Listeria monocytogenes} in 
eukaryotic cells \cite{hotton} and the Barkley-Kevrekidis model for the meander
of spiral waves in excitable media such as the BZ reaction \cite{barkley}.

\medskip
\noindent
{\bf Example 3 - Actin based motility} \\ {\it L. monocytogenes} transport 
themselves in eukaryotic cells by catalyzing the polymerization of the 
cytoskeletal protein actin.  This method of propulsion can result in a 
bacterium following a complicated path within the cytosol at a fairly constant 
speed.  The curvature of the paths tends to vary periodically with time so the 
curvature and speed have a common minimal period and we can apply the theory 
from section \ref{S:kappabreve}.

   Recall from section \ref{S:kappabreve} that the velocity's orientation is
$\varphi(t) = \; \overline{\kappa}\; t + \widetilde{\varphi}(t)$ where   
$\widetilde{\varphi}(t)$ is periodic.  In the Shenoy-Rutenberg model 
$\widetilde{\varphi}(t) = (\Omega/\omega_0) \sin(\omega_0 t)$ where $\omega_0$ 
is the angular frequency of the bacterium's spin about its long axis and 
$\Omega$ is a monotonic function of the distance of the effective propulsive 
force from the long axis.  In the model the speed of the bacterium is the 
constant $v_0$.  The common minimal period of the curvature and speed is 
$T=2\pi/\omega_0$ and so $\overline{\kappa} = \omega_0 \k$.  The points of 
maximal curvature occur for $t \in (2\pi/\omega_0) {\bf Z}$ and the points of 
minimal curvature occur for $t \in (\pi/\omega_0) + (2\pi/\omega_0) {\bf Z}$.

   Since $(\Omega/\omega_0) \sin(\omega_0 t)$ is an odd function the full
image of the curve has reflectional symmetry.  Joining an arc over a 
half-period with its reflected image gives a periodic arc.  For integral 
$\k$ the rest of the curve can be obtained by translating the periodic arc.
For non-integral $\k$ the rest of the curve can be obtained by rotating the 
periodic arc about $(\overline{x}, \overline{y})^T$.

\begin{figure}[ht]
\centering
\includegraphics[scale=0.25]{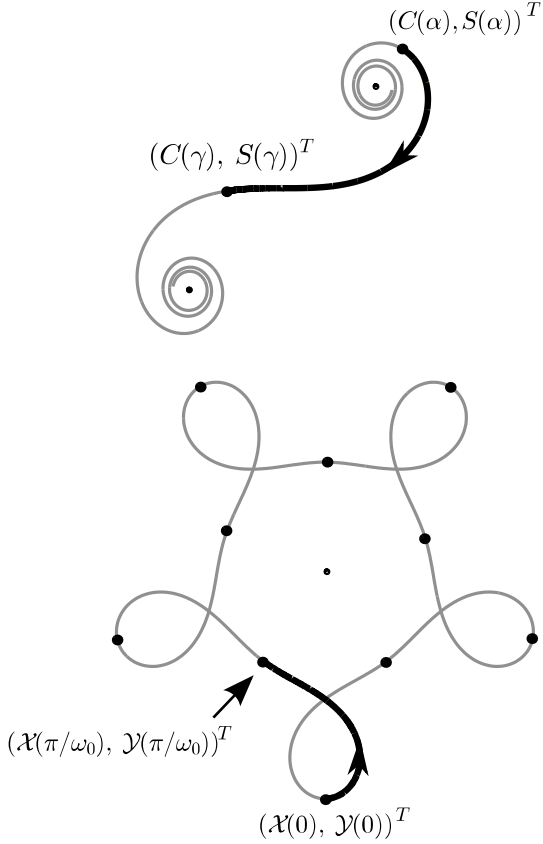}
\caption{Generating a spiral easement approximation for an $(x,y)$-curve.  
Above: A clothoid.  It is symmetrical under a half turn and the center is its 
unique inflection point.  Below: The approximating $(\mathcal{X},\, 
\mathcal{Y})$-curve.  The two arcs highlighted in black are geometrically 
similar to each other.  The $(\mathcal{X},\, \mathcal{Y})$-curve can be 
obtained by successively reflecting its black arc about the five mirror lines
shown in figure \ref{F:approx}.}
\label{F:clothoid}
\end{figure}

    When the initial direction is horizontal, \ie $\varphi(0) = 0$, and the 
starting point is at the origin, \ie $(x(0),y(0))^T=(0,0)^T$, the time 
parameterization for the curve is 
\begin{small}
\begin{equation} 
\label{E:timeparameterization}
\begin{pmatrix}
 x(t) \\ y(t) 
\end{pmatrix} = v_0 \int_0^t
\begin{pmatrix}
\cos(\omega_0 \k \tau + (\Omega/\omega_0) \, \sin(\omega_0 \, \tau)) \\
\sin(\omega_0 \k \tau + (\Omega/\omega_0) \, \sin(\omega_0 \, \tau)) 
\end{pmatrix}
d\tau
\end{equation}
\end{small}
There is no closed form for this integral in terms of elementary functions but 
we can obtain a closed form that accurately approximates it by borrowing a 
technique from civil engineering known as spiral easement.  This technique 
varies the curvature of roads and train tracks in a piecewise linear fashion.  
For the Shenoy-Rutenberg model we approximate the curvature with the triangular
wave form
\begin{equation*}
\kappa(t) \approx \frac{\omega_0\, \k}{v_0} + \frac{\Omega}{v_0} \left( 
\frac{8}{\pi^2} \arcsin(\cos(\omega_0 t)) \right)
\end{equation*}
This gives a piecewise quadratic approximation for $\widetilde{\varphi}(t)$:
\begin{small}
\begin{equation*}
\frac{\Omega}{\omega_0} \; \left(\frac{16}{\pi^2} \; 
\arcsin(\cos(\omega_0\,  t/2)) 
\arcsin(\sin(\omega_0\,  t/2))\right)
\end{equation*}
\end{small}
The rotational symmetry of the curve is unaffected by this approximation since
only the $\widetilde{\varphi}(t)$ term in $\varphi(t)$ is altered and the 
reflectional symmetry is unaffected since $\widetilde{\varphi}(t)$ remains an 
even function.  Because of the symmetry the error varies periodically and so 
remains bounded.  The quantity
\begin{small}
\begin{eqnarray*} 
\left| \frac{16}{\pi^2} \; \arcsin(\cos(\omega_0 t/2)) 
\arcsin(\sin(\omega_0 t/2)) - \sin(\omega_0 t) \right|
\end{eqnarray*} 
\end{small}
is never more than $3.21^o$ at any time.  Figure \ref{F:approx} shows an
example of a spiral easement approximation for the parametrized curve in
equation \eqref{E:timeparameterization}

   We let $(\mathcal{X}(t),\, \mathcal{Y}(t))^T$ denote the time 
parameterization for the approximating curve to \eqref{E:timeparameterization}.
On the interval $[0, \pi/\omega_0]$ the curvature can be written as the linear 
polynomial $2\beta(\alpha - \beta t)$ 
where 
\begin{equation*}
\alpha = \frac{\pi \omega_0 \k + 4\Omega}{2\pi \beta v_0} \qquad \qquad
\beta = \frac{2}{\pi} \sqrt{\frac{\omega_0 \Omega}{v_0}}
\end{equation*}
Integrating over a sub-interval $[0,t] \subseteq [0, \pi/\omega_0]$ gives us a 
closed form time parameterization for an arc of the $(\mathcal{X}(t),\, 
\mathcal{Y}(t))$-curve,
\begin{small}
\begin{equation}
\label{E:clothoid}
\begin{pmatrix}
 \mathcal{X}(t) \\ \mathcal{Y}(t) 
\end{pmatrix} = \frac{v_0}{\beta} R_{(\alpha^2)}
\begin{pmatrix}
 -C(\alpha-\beta t) + C(\alpha) \\
~~S(\alpha-\beta t) - S(\alpha) 
\end{pmatrix}
\end{equation}
\end{small}
where $C(t)$, $S(t)$ are the Fresnel trigonometric functions.  

\begin{figure}[ht]
\centering
\includegraphics[scale=0.55]{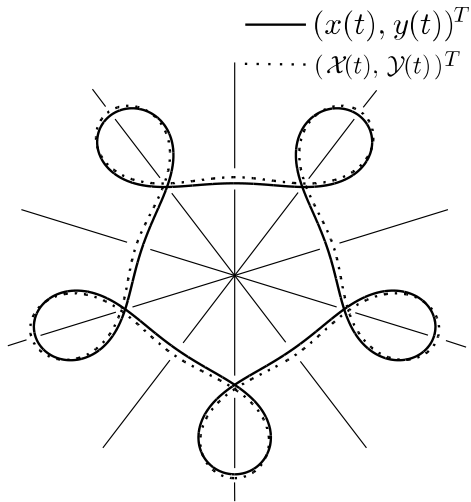}
\caption{An $(x,y)$-curve with $(\k,\Omega/\omega_0)=(4/5,1)$ along with its 
approximating $(\mathcal{X},\mathcal{Y})$-curve translated to have the same 
center.  The curves' five mirror lines are also displayed.}
\label{F:approx}
\end{figure}

   The planar curve $t \mapsto (C(t),S(t))^T$ is known as a clothoid\footnote{
It is also known as Euler's spiral and as Cornu's spiral.}.  It is shown at the
top of figure \ref{F:clothoid}.  It has unit speed so $t$ is the arc length
from it center, $(0,0)^T$ to $(C(t),S(t))^T$.  The curvature at $(C(t),S(t))^T$
is $2t$ so every possible curvature occurs at exactly one point of the 
clothoid.  Since the curvature is monotonic the clothoid does not intersect 
itself and since the curvature is unbounded in the positive and negative
directions the clothoid spirals around two points.

    The right hand side of equation \eqref{E:clothoid} is the application of a 
odd similarity transformation to an arc of the clothoid.  This approximation 
technique amounts to taking the extremal curvatures of the $(x,y)^T$-curve, 
determining the two points on the clothoid where these extremal curvatures
occur, and applying an odd similarity to the arc in the clothoid connecting the
points of extremal curvature as shown in figure \ref{F:clothoid}.  The rest of 
the $(\mathcal{X},\, \mathcal{Y})$-curve is obtained by the action of the 
symmetry group of the $(\mathcal{X},\, \mathcal{Y})$-curve.  The approximation 
can be further refined by translating the $(\mathcal{X}, \mathcal{Y})$-curve so 
that it has the same center as the $(x,y)$-curve (see figure \ref{F:approx}).

   The value of $\k$ determines the rotational symmetry of the curve.  The 
value of $\Omega/\omega_0$ determines the length of the clothoid arc used
in approximating the curve.  The effect of $\Omega/\omega_0$ for fixed $\k$ is 
perhaps best illustrated in the $\k=1/2$ case (see figure 
\ref{F:kappabreveOneHalf}) because there are fewer self-intersections to deal 
with.

   As $\Omega/\omega_0$ varies the points of minimal curvature oscillate in 
unison along mirror lines and the points of maximal curvature oscillate in 
unison along mirror lines.  For small $\Omega/\omega_0$ the curve has an oval 
shape with the points of minimal curvature closer to the center than the points
of maximal curvature.  As $\Omega/\omega_0$ increases the points of minimal 
curvature move toward the center and inflection points appear.  Next the points
of minimal curvature pass through the center together and then move outward.  
Eventually the points of minimal curvature go far from the center while the 
points of maximal curvature go near the center.  The points of maximal 
curvature pass through the center together and then proceed outward.  
Afterwards the points of minimal curvature pass back through the center again. 
The process is reminiscent of a loom except the curve becomes wound up around 
four points (when $\k =1/2$) instead of being woven.

\begin{figure}[ht]
\centering
\includegraphics[scale=0.225]{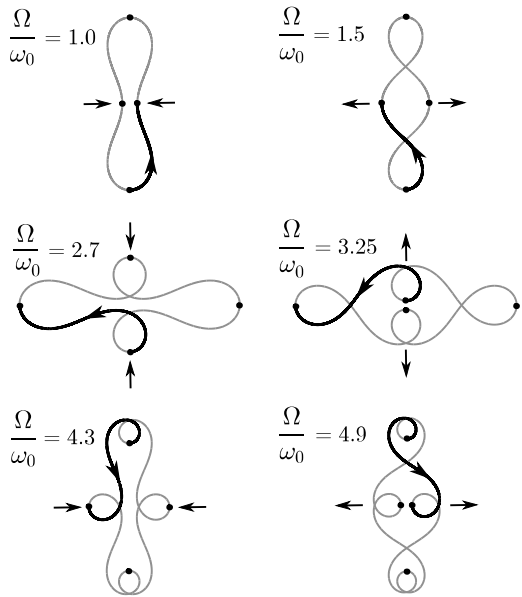}
\caption{$(\mathcal{X},\mathcal{Y})$-curves with fixed $\k = 1/2$ and 
increasing $\Omega/\omega_0$.  Each row shows a pair of points with extremal 
curvature passing through each other at the center thereby introducing a pair 
of crossing points which persist as $\Omega/\omega_0$ continues to increase.}  
\label{F:kappabreveOneHalf}
\end{figure}

    More generally, for $\k \notin {\bf Z}$, the points of maximal curvature 
coincide with the center if and only if $(\mathcal{X}(0), \mathcal{Y}(0))^T = 
(\overline{\mathcal{X}},\overline{\mathcal{Y}})^T$ while the points of minimal 
curvature coincide with the center if and only if $(\mathcal{X}(\pi/\omega_0), 
\mathcal{Y}(\pi/\omega_0))^T = (\overline{\mathcal{X}},
\overline{\mathcal{Y}})^T$.  If the points of minimal curvature coincide with 
the center then 
\begin{small}
\begin{equation}
\label{E:Fcosine}
\cos(\alpha^2)\left( C(\gamma) - C(\alpha)\right)
+\sin(\alpha^2)\left( S(\gamma) - S(\alpha)\right) = 0 
\end{equation}
\end{small}
where $\gamma=\alpha-\beta\pi/\omega_0$.  This gives us a condition which must
hold between $\k$ and $\Omega/\omega_0$.  The same condition holds if the 
points of maximal curvature coincide with the center except $\k$ is replaced
with $-\k$.  These two conditions determine two sets of curves in the 
$(\k,\; \Omega/\omega_0)$ parameter plane which are shown in figure 
\ref{F:parameterplane}.  

   These two sets of curves only intersect at integer values of $\k$. It turns  
out in these cases that the congruence $\mathcal{G}_{\k, T}$ reduces to 
the identity map, that the image of the $(\mathcal{X}, \mathcal{Y})$-curve is 
closed, and that $\k$ is its turning number.

\begin{figure}[ht]
\centering
\includegraphics[scale=0.60]{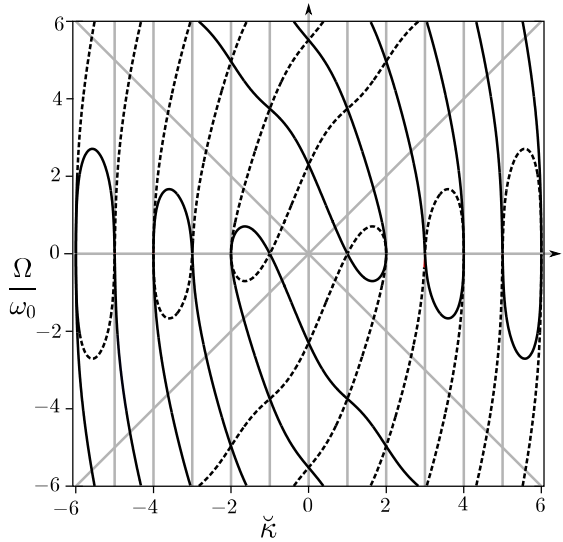}
\caption{The parameter plane for equation \eqref{E:clothoid} and the curves
defined by \eqref{E:Fcosine}.  The solid black curves correspond to 
$(\mathcal{X},\, \mathcal{Y})$-curves whose points of minimal curvature 
coincide with the center.  The dotted black curves correspond $(\mathcal{X}, \, 
\mathcal{Y})$-curves whose points of maximal curvature coincide with the 
center.  The solid black and dotted black curves only intersect on the vertical 
gray lines which correspond to integer values for $\k$.  The gray diagonal 
lines correspond to the appearance of inflection points in the $(\mathcal{X},
\mathcal{Y})$-curves.}
\label{F:parameterplane}
\end{figure}

\medskip
\noindent
{\bf Example 4 - Spiral tip meander} \\ 
  Barkley's model \cite{barkley} for the spiral wave tip meander can be written 
as the ordinary differential equation 
\begin{equation*}
\begin{pmatrix}
\dot{x}       \\
\dot{y}       \\
\dot{\varphi} \\
\dot{v}       \\
\dot{w}       
\end{pmatrix}
=
\begin{pmatrix}
 v \cos(\varphi)                                   \\
 v \sin(\varphi)                                   \\
 \gamma_0 \, w                                        \\
 v \; (-1/4 + (10/3) v^2 + \alpha_2 \, w^2 - v^4)  \\
 w \; (-1 + v^2 - w^2)
\end{pmatrix} 
\end{equation*}
Here we have replaced the variable name `$s$' in Barkley's equations (3) and 
(4) with the variable name `$v$' in keeping with the convention in this article 
that $s$ stands for arc length and $v$ stands for speed.

  From equation \eqref{E:kappabreve} the total curvature per periodic arc is  
\begin{equation*}
\k \approx \frac{\sqrt{7} \, \nu}{\pi} \int_0^{\sqrt{14} \pi/7} w(\tau) \, 
d\tau
\end{equation*}
where $\nu = \gamma_0/\sqrt{28}$.  The isogonal curves for $\k = 1,2,3,4$ are 
shown in figure \ref{F:comparison} where $\mu = -(\alpha_2+5)/5$.  

\begin{figure}[ht]
\centering
\includegraphics[scale=0.55]{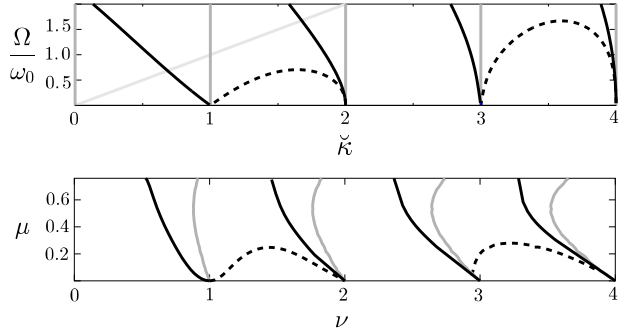}
\caption{(Above) Blow up of a region in figure \ref{F:parameterplane}.  (Below)
The $(\nu, \mu)$ parameter space for Barkley's ODE (compare to figure 3 in 
\cite{barkley}).  Curves in the two parameter spaces are depicted in the same 
manner as in figure \ref{F:parameterplane}.}
\label{F:comparison}
\end{figure}

\medskip
\noindent
\section{Conclusion}

\end{document}